\documentclass[UKenglish,numberwithinsect]{lipics-v2016}

\usepackage{graphicx}
\usepackage{todonotes}
\usepackage{mathrsfs}
\usepackage{hyperref}
\usepackage{paralist}

\usepackage{framed}
\renewcommand{\vec}[1]{\boldsymbol{#1}}

\usepackage[noend]{algorithmic}
\usepackage[boxed]{algorithm}
\usepackage{xspace}

\newcommand{\algorithmiccommentt}[1]{\colorbox{black!10}{#1}}
\newcommand{\LineIf}[2]{ \STATE \algorithmicif\ {#1}\ \algorithmicthen\ {#2} }
\newcommand{\LineIfElse}[3]{ \STATE \algorithmicif\ {#1}\ \algorithmicthen\ {#2} \algorithmicelse\ {#3} }

\newcommand{\setcover}{\textsc{Set Cover}\xspace}
\newcommand{\setpartition}{\textsc{Set Partition}\xspace}
\newcommand{\subsetsum}{\textsc{Subset Sum}}

\newcommand{\cL}{\mathcal{L}}
\newcommand{\cR}{\mathcal{R}}

\newcommand{\cW}{\mathcal{W}}
\newcommand{\cZ}{\mathcal{Z}}

\newcommand{\wt}{\mathsf{wt}}

\DeclareMathOperator{\poly}{poly}

\newtheorem{observation}[theorem]{Observation}

\newtheorem{openproblem}{Open Problem}

\numberwithin{equation}{section}

\makeatletter

\newtheorem*{rep@theorem}{\rep@title}
\newcommand{\newreptheorem}[2]{%
\newenvironment{rep#1}[1]{%
 \def\rep@title{#2 \ref{##1}}%
 \begin{rep@theorem}[restated]}%
 {\end{rep@theorem}}}
\makeatother

\newreptheorem{theorem}{Theorem}

\newcommand{\yes}{\ensuremath{\mathbf{yes}}\xspace}
\newcommand{\no}{\ensuremath{\mathbf{no}}\xspace}

\newcommand{\solsize}{\ensuremath{s}}

\long\def\comment#1{}

\begin{document}

\EventEditors{Piotr Sankowski and Christos Zaroliagis} 
\EventNoEds{2} 
\EventLongTitle{24th Annual European Symposium on Algorithms (ESA 2016)}
\EventShortTitle{ESA 2016} 
\EventAcronym{ESA} 
\EventYear{2016} 
\EventDate{August 22--24, 2016} 
\EventLocation{Aarhus, Denmark} 
\EventLogo{} 
\SeriesVolume{57} 
\ArticleNo{[79]}

\Copyright{Jesper Nederlof}%mandatory, please use full first names. LIPIcs license is "CC-BY";  http://creativecommons.org/licenses/by/3.0/

% mandatory: Please choose ACM 1998 classifications from http://www.acm.org/about/class/ccs98-html . E.g., cite as "F.1.1 Models of Computation". 
% mandatory: Please provide 1-5 keywords
% Author macros::end %%%%%%%%%%%%%%%%%%%%%%%%%%%%%%%%%%%%%%%%%%%%%%%%%

\title{Finding Large Set Covers Faster via the Representation Method\footnote{Funded by the NWO VENI project 639.021.438. This work was partly done while the author was visiting the Simons Institute for the Theory of Computing during the program `Fine-Grained Complexity and Algorithm Design' in the fall of 2015.}}
%\author{Jesper Nederlof\footnote{Department of Mathematics and Computer Science, Technical University of Eindhoven, The Netherlands, \texttt{j.nederlof@tue.nl}. }}
\author{Jesper Nederlof}
\authorrunning{Jesper Nederlof}
\affil[1]{Department of Mathematics and Computer Science, Eindhoven University of Technology, The Netherlands. \texttt{j.nederlof@tue.nl}}

\maketitle

\begin{abstract}
The worst-case fastest known algorithm for the Set Cover problem on universes with $n$ elements still essentially is the simple $O^*(2^n)$-time dynamic programming algorithm, and no non-trivial consequences of an $O^*(1.01^n)$-time algorithm are known. Motivated by this chasm, we study the following natural question: Which instances of Set Cover \emph{can} we solve faster than the simple dynamic programming algorithm? Specifically, we give a Monte Carlo algorithm that determines the existence of a set cover of size $\sigma n$ in $O^*(2^{(1-\Omega(\sigma^4))n})$ time. Our approach is also applicable to Set Cover instances with exponentially many sets: By reducing the task of finding the chromatic number $\chi(G)$ of a given $n$-vertex graph $G$ to Set Cover in the natural way, we show there is an $O^*(2^{(1-\Omega(\sigma^4))n})$-time randomized algorithm that given integer $\solsize=\sigma n$, outputs NO if $\chi(G) > \solsize$ and YES with constant probability if $\chi(G)\leq \solsize-1$.

On a high level, our results are inspired by the `representation method' of Howgrave-Graham and Joux~[EUROCRYPT'10] and obtained by only evaluating a randomly sampled subset of the table entries of a dynamic programming algorithm. 
\end{abstract}

\section{Introduction}
\label{sec:intro}
The \setcover problem is, after determining satisfiability of CNF formulas or Boolean circuits, one of the canonical NP-complete problems. It not only directly models many applications in practical settings, but also algorithms for it routinely are used as tools for theoretical algorithmic results (e.g.,~\cite{Dantsin200269}). It is a problem `whose study has led to the development of fundamental techniques for the entire field' of approximation algorithms.\footnote{As the Wikipedia page on Set Cover quotes the textbook by Vazirani~\cite[p15]{Vazirani:2001:AA:500776}.} However, the exact exponential time complexity of \setcover is still somewhat mysterious: We know algorithms need to use super-polynomial time assuming $P\neq NP$ and (denoting $n$ for the universe size) $O^*(2^{\Omega(n)})$ time assuming the Exponential Time Hypothesis, but how large the exponential should be is not clear. In particular, no non-trivial consequences of an $O^*(1.01^n)$-time algorithm are currently known.

Even though it is one of the canonical NP-complete problems, the amount of studies of exact algorithms for \setcover pales in comparison with the amount of literature on exact algorithms for \textsc{CNF-Sat}: Many works focus on finding $O^*(c^n)$-time algorithms for $c<2$ for \textsc{CNF-Sat} on $n$-variable CNF-formulas in special cases such as, among others, bounded clause width~\cite{DBLP:journals/algorithmica/Schoning02,Dantsin200269,DBLP:conf/soda/ChanW16}, bounded clause density~\cite{DBLP:conf/coco/CalabroIP06,DBLP:journals/corr/ImpagliazzoLPS14} or few projections~\cite{DBLP:conf/iwpec/KaskiKN12,Paulusma2015,saether}. Improved exponential time algorithms for special cases of problems other than \textsc{CNF-Sat} were also studied for e.g. \textsc{Graph Coloring} or \textsc{Traveling Salesman} on graphs bounded degree/average degree~\cite{DBLP:journals/mst/BjorklundHKK10,DBLP:journals/talg/BjorklundHKK12,DBLP:journals/iandc/CyganP15, DBLP:conf/icalp/GolovnevKM14}.

In this paper we are interested in the exponential time complexity of \setcover, and study which properties are sufficient to have improved exponential time algorithms. Our interest in finding faster exponential time algorithms for \setcover does not only stem from it being a canonical NP-complete problem, but also from its unclear relation with \textsc{CNF-Sat}. Intriguingly, on one hand \setcover has some similarities with the \textsc{CNF-Sat}: \begin{inparaenum} \item Both problems take an (annotated) hypergraph as input \item The improvability of the worst-case complexity of \textsc{CNF-Sat} is essentially equivalent to the improvability of the worst-case complexity of \textsc{Hitting Set}~\cite{DBLP:conf/coco/CyganDLMNOPSW12}, which is just a reparametrization\footnote{One way of stating \textsc{Hitting Set} in this context, is that we have an instance of the \setcover problem but aim to find an $O^*(2^{(1-\Omega(1))m})$ time algorithm, where $m$ denotes the number of sets.} of \setcover.\end{inparaenum}\ But, on the other hand the problems are quite different to our understanding: \begin{inparaenum} \item Most algorithms for \setcover use dynamic programming or some variant of inclusion exclusion, while most algorithms for \textsc{CNF-Sat} are based on branching \item No connection between the exponential time complexities of both problems is known (see~\cite{DBLP:conf/coco/CyganDLMNOPSW12})\end{inparaenum}. One hope would be that a better understanding of the exact complexity of \setcover might shed more light on this unclarity.
Moreover, Cygan et al.~\cite{DBLP:conf/coco/CyganDLMNOPSW12} also show that if we would like to improve the run time $O^*(f(k))$ of several parameterized algorithms to $O^*(f(k)^{1-\Omega(1)})$, we first need to find an $O^*(2^{(1-\Omega(1))n})$-time algorithm for \setcover. These parameterized algorithms include the classic algorithm for \subsetsum{}, as well as more recent algorithms for \textsc{Connected Vertex Cover} and \textsc{Steiner Tree}.

\subparagraph{Relevant previous work} The algorithmic results on \setcover that are the most relevant to our work are as follows: The folklore dynamic programming algorithm runs in $O^*(2^n)$ time. A notable special case of \setcover that can be solved in $O^*(2^{(1-\Omega(1))n})$ time is due to Koivisto~\cite{DBLP:conf/iwpec/Koivisto09}: He gives an algorithm that runs in time $O^*(2^{(1-\frac{1}{O(r)})n})$-time algorithm if all sets are at most of size $r$. Bj\"orklund et al.~\cite{DBLP:journals/siamcomp/BjorklundHK09} show that the problem can be solved in $2^n\poly(n)$ time (which is faster if the number of sets is exponentially large in $n$). Bj\"orklund et al.~\cite{DBLP:journals/corr/abs-1007-1161} give a randomized algorithm that assumes all sets are of size $q$ and determines whether there exist $p$ pairwise disjoint sets in $O^*(2^{(1-\epsilon)pq})$ time where $\epsilon >0$ depends on $q$.

\subparagraph{Our Main Results}

We investigate what are sufficient structural properties of instances of \setcover, and the closely related \setpartition (in which the picked sets need to be disjoint), problems to be solvable in time significantly faster than the currently known algorithms. We will outline our main results now:

\begin{theorem}\label{thm:largesol}
There is a Monte Carlo algorithm that takes an instance of \setcover on $n$ elements and $m$ sets and an integer $\solsize$ as input and determines whether there exists a set cover of size $s$ in $O(2^{(1-\Omega(\sigma^4))n}m)$ time, where $\sigma=s/n$.
\end{theorem}

We remark that this generalizes the result of Koivisto~\cite{DBLP:conf/iwpec/Koivisto09} in the sense that it solves a larger class of instances in $O^*(2^{(1-\Omega(1))n})$ time: If all set sizes are bounded by a constant $r$, a set partition needs to consist of at least $n/r$ sets and Theorem~\ref{thm:largesol} applies with $\sigma = 1/r$ (although this gives a slower algorithm than Koivisto's in this special case). Moreover, it seems hard to extend the approach of Koivisto to our more general setting.

The second result demonstrates that our techniques are also applicable to \setcover instances with exponentially many sets, a canonical example of which being graph coloring:

\begin{theorem}\label{thm:col}
There is a randomized algorithm that given graph $G$ and integer $\solsize=\sigma n$, in $O^*(2^{(1-\Omega(\sigma^4))n})$ time outputs
\yes with constant probability, if $\chi(G) < \solsize$, and \no, if $\chi(G) > \solsize$.
\end{theorem}

\subparagraph{Representation method for Set Cover} We feel the main technique used in this paper is equally interesting as the result, and will therefore elaborate on its origin here. Our technique is on a high level inspired by the following simple observation ingeniously used by Howgrave-Graham and Joux~\cite{DBLP:conf/eurocrypt/Howgrave-GrahamJ10}: Suppose $\mathcal{S} \subseteq 2^{[m]}$ is a set of solutions implicitly given and we seek for a solution $X \in \mathcal{S}$ with $|X|=s$ by listing all sets of $\binom{[m]}{s/2}$ and performing pairwise checks to see which two combine to an element of $\mathcal{S}$. Then we can restrict our search in various ways since there will be as many as $\binom{s}{s/2}$ pairs guiding us to $X$. In~\cite{DBLP:conf/eurocrypt/Howgrave-GrahamJ10} and all subsequent works (including~\cite{DBLP:conf/eurocrypt/BeckerCJ11,eurocrypt-2012-24271,DBLP:conf/stacs/AustrinKKN15,DBLP:journals/corr/AustrinKKN15}), this idea was used to speed up `meet-in-the-middle attacks' (also called `birthday attacks'~\cite[Chapter 6]{Joux:2009:AC:1642719}). We will refer to uses of this idea as the `representation method' since it crucially relies on the fact that $X$ has many representations as pairs. To indicate the power of this technique in the context of \setcover and \setpartition we show that without changes it already gives an $O^*(2^{0.3399m})$-time Monte Carlo algorithm for the \setpartition problem with $m$ sets, and even for a more general linear satisfiability problem on $m$ variables. For the latter problem this improves the $O^*(2^{m/2})$ time algorithm based on the meet-in-the-middle attack that was the fastest known before.

At first sight the representation method seemed to be inherently only useful for improving algorithms based on the meet-in-the-middle attack. However, the main conceptual contribution of this work is to show that it is also useful in other settings, or at least for improving the dynamic programming algorithm for the \setcover and \setpartition problems if the solution size is large. On a high level, we show this as follows in the case of \setpartition:\footnote{The algorithm for \setcover actually reduces to \setpartition.} for a subset $W$ of the elements of the \setpartition instance, define $T[W]$ to be the minimum number of disjoint sets needed to cover all elements of $W$. Stated slightly oversimplified, we argue that if a minimal set partition of size $s$ is large, we have that $T[W]+T[[n] \setminus W]=s$ for $\binom{s}{s/2}$ sets $W$ with $|W|$ close to $n/2$. To relate this to later sections, let us remark we refer to such a set $W$ as a \emph{witness halve}. Subsequently, we exploit the presence of many witness halves by using a dynamic programming algorithm that samples a set of the subsets with size close to $n/2$ and only evaluates table entries from this sample plus the table entries required to compute the table entries from the sample.

\subparagraph{Organization}

This paper is organized as follows: In Section~\ref{sec:prel}, we recall preliminaries and introduce notation. In Section~\ref{sec:vars}, we discuss new observations and basic results that we feel are useful for developing a better understanding of the complexity of \setcover with respect to several structural properties of instances. In Section~\ref{sec:manyhalves} we formally present the notion of witness halves and prepare tools for exploiting the existence of many witness halves. In Section~\ref{sec:mainproofs} we prove our main results and in Section~\ref{sec:concremark} we suggest further research.

\section{Preliminaries and Notation}\label{sec:prel}
For a real number $x$, $|x|$ denotes the absolute value of $x$. For a Boolean predicate $p$, we let $[p]$ denote $1$ if $p$ is true and $0$ otherwise. On the other hand, if $p$ is an integer we let $[p]$ denote $\{1,\ldots,p\}$. As usual, $\mathbb{N}$ denotes all positive integers. Running times of algorithms are often stated using $O^*(\cdot)$ notation which suppresses factors polynomial in the input size. To avoid superscript, we sometimes use $\exp(x)$ to denote $e^x$. We denote $\lg$ for the base-$2$ logarithm. If $G=(V,E)$ and $v \in V$ we denote $N(v)= \{w \in V: (v,w)\in E\}$ and for $X \subseteq V$ we extend this notation to $N(X)=\bigcup_{v \in X}N(v)$. For reals $a,b>0$ we let $a \pm b$ denote the interval $[a-b,a+b]$. A false positive (negative) of an algorithm is an instance on which it incorrectly outputs YES (respectively, NO). In this work we call an algorithm Monte Carlo if it has no false positives and if any instance is a false negative with probability at most $1/4$. We denote vectors with boldface for clarity. For a real number $x \in [0,1]$, $h(x)=  -x \lg x - (1 - x) \lg (1 - x)$ denotes the binary entropy of $x$, where $0 \lg 0$ should be thought of as $0$. It is well known that $\binom{b}{a} \leq 2^{h(a/b)b}$ (and this can for example be proved using Stirling's approximation). It is easy to see from the definition that $h(\cdot)$ is symmetric in the sense that $h(x)=h(1-x)$. 
\begin{lemma}
The following can be verified using standard calculus:
  \label{fact:bounds}
	\begin{enumerate}
		\item $h(1/2 - x) = h(1/2+x) \leq 1 - x^2$ for all $x \in (0,1/2)$,
		\item $h(x) \leq x \lg( 4/x )$ for all $x \in (0,1)$,
		\item $(1-1/n)^n \leq 1/e$.
	\end{enumerate}
\end{lemma}
\begin{lemma}[Hoeffding bound~\cite{hoeffding}]\label{lem:hoeff}
If $X_1,\ldots,X_s$ are independent, $Y = \sum_{i=1}^s X_i$ and $a_i \leq X_i \leq b_i$ for $i=1,\ldots,s$ then
$\Pr[|Y- \mathbb{E}[Y]| \geq t] \leq 2 \cdot \exp\left(\frac{-2t^2}{\sum_{i=1}^s (b_i-a_i)^2} \right)$.
\end{lemma}

\subparagraph{Set Cover / Set Partition}
In the \setcover problem we are given a bipartite graph $G=(F \dot{\cup} U, E)$ (where $F$ and $U$ shorthand `Family' and `Universe' respectively), together with an integer $\solsize$ and the task is to determine whether there exists a \emph{solution} $S \subseteq F$ such that $N(S)=U$ and $|S| \leq \solsize$. In the \setpartition problem we are given the same input as in the \setcover problem, but we are set to determine whether there exists $S \subseteq F$ with $N(S)=U$, $|S| = \solsize$ and additionally $N(f) \cap N(f')=\emptyset$ for every $f,f' \in S$ with $f\neq f'$. We will refer to solutions of both problems as set covers and set partitions.

Throughout this paper, we let $n,m$ respectively denote $|U|$ and $|F|$, and refer to instances of \setcover or \setpartition as $(n,m,\solsize)$-instances to quantify their parameters.
Since this work concerns \setcover or \setpartition with large solutions we record the following basic observation that follows by constructing for each\footnote{For \setpartition only do this for $c$-tuples $(f_1,\ldots,f_c)$ with $N(f_i)$ disjoint.} $c$-tuple $t=(f_1,\ldots,f_c) \in F^c$ of sets in the original instance a set $f^t$ with $N(f^t)=\bigcup_{i=1}^t f_i$ in the output instance:
\begin{observation}[\cite{DBLP:conf/coco/CyganDLMNOPSW12}]\label{obs:redsolsize}
There is a polynomial time algorithm that takes a constant $c \geq 1$ dividing $\solsize$, and a $(n,m,\solsize)$-instance of \setcover (resp. \setpartition) as input and outputs an equivalent $(n,m^c,\solsize/c)$-instance of \setcover (resp. \setpartition).
\end{observation}

Often it will be useful dispense with linear sized sets. To this end, the following can be achieved by simply iterating over all $f \in F$ with $|N(f)| \geq \epsilon n$ and checking for each such set whether there is a solution containing it using the $2^{n}\poly(n)$ algorithm for \setcover~\cite{DBLP:journals/siamcomp/BjorklundHK09}.
\begin{observation}\label{obs:nolinear}
There is an algorithm that, given a real number $\epsilon >0$, takes an $(n,m,\solsize)$-instance of \setcover as input and outputs an equivalent $(n,m',\solsize)$-instance with $m' \leq m$ satisfying $|N(f)| \leq \epsilon n$ for every $f \in F$. The algorithm runs in $O(m2^{(1-\epsilon)n}\poly(n))$ time.
\end{observation}

As we will see in Theorem~\ref{thm:subtle}, it makes a difference in the \setpartition problem whether empty sets are allowed since we need to find a set partition of size exactly $s$. To exclude such sets, we will simply say that an instance is `without empty sets'.

\section{Observations and Basic Results on Set Cover and Set Partition.}\label{sec:vars}

To improve our understanding of which properties of instances of \setcover and \setpartition allow faster algorithms, and which techniques are useful for obtaining such faster algorithms, we will record some observations and basic results in this section. To stress that the proof techniques in this section are \emph{not our main technical contribution}, we postpone all proofs to Appendix~\ref{app:A}.

We prefer to state our results in terms of \setcover because it is slightly more natural and common, but since \setpartition often is easier to deal with for our purposes we will sometimes use the following easy reduction, all of whose steps are contained in~\cite{DBLP:conf/coco/CyganDLMNOPSW12}:
\begin{theorem}\label{thm:redcovpart}
There is an algorithm that, given a real $0 <\epsilon < 1/2$, takes an $(n,m,\solsize)$-instance of \setcover as input and outputs an equivalent $(n,m',\solsize)$-instance of \setpartition with $m' \leq m2^{\epsilon n}$ sets in time $O(m2^{(1-\epsilon)n})$.
\end{theorem}

For completeness, we show that in fact \setcover and \setpartition are equivalent with respect to being solvable in time $O^*(2^{(1-\Omega(1))n})$. This was never stated in print to the best of our knowledge, but the proof uses standard ideas and is found in Appendix~\ref{app:covparteq}.
\begin{theorem}\label{thm:covparteq}
For some $\epsilon >0$ there is an $O^*(2^{(1-\epsilon)n})$ time algorithm for \setcover if and only if for some $\epsilon' >0$ there is an $O^*(2^{(1-\epsilon')n})$ time algorithm for \setpartition.
\end{theorem}

The following natural result is a rather direct consequence of a paper by Koivisto~\cite{DBLP:conf/iwpec/Koivisto09}. It reveals some more similarity with the \textsc{$k$-CNF-Sat} problem: Koivisto shows\footnote{Koivisto only showed this for \setpartition, but the straightforward reductions in this section carry this result over to \setcover.} that for maximum set size $r$, \setcover can be solved in $O^*(2^{(1-\Omega(\frac{1}{r}))n})$ which is analogous to \textsc{$k$-CNF-Sat} being in $O^*(2^{(1-\Omega(\frac{1}{k}))n})$ time~\cite{DBLP:journals/algorithmica/Schoning02,Dantsin200269,DBLP:conf/soda/ChanW16}, and similarly the following result is the counterpart of $O^*(2^{(1-\Omega(\frac{1}{\delta}))n})$-time algorithms for CNF-formula's of density $\delta$ (i.e. at most $\delta n$ clauses)~\cite{DBLP:conf/coco/CalabroIP06,DBLP:journals/corr/ImpagliazzoLPS14}. Again, this result was never explicitly stated in print to the best of our knowledge, and therefore is proved in Appendix~\ref{app:fewsets}.

\begin{theorem}\label{thm:fewsets}
There is an algorithm solving $(n,m,\solsize)$-instances of \setcover or \setpartition in time $m\cdot\mathsf{poly}(n)2^{n-\frac{n}{O(\lg(m/n))}}$.
\end{theorem}

Relevant to our work is the following subtlety on solution sizes in \setpartition. It shows that for \setpartition with empty sets, finding large solutions is as hard as the general case. The proof is postponed to Appendix~{\ref{app:subtle}}.

\begin{theorem}\label{thm:subtle}
Suppose there exist $0<\epsilon_1,\epsilon_2 <1/2$ and an algorithm solving $(n,m,\epsilon_1n)$-instances of \setpartition in time $O^*(2^{(1-\epsilon_2)n})$. Then there exists an $O^*(2^{(1-\epsilon_2/2)n})$-time algorithm for \setpartition.
\end{theorem}

Finally, it is insightful to see how well the representation method performs on the \setpartition problem with few sets (e.g., we consider running times of the $O^*(2^{m})$, where $m$ is the number of sets). A straightforward approach of the meet-in-the-middle attack leads directly to an $O^*(2^{m/2})$ time algorithm. We show that the representation method combined with the analysis of~\cite{DBLP:journals/corr/AustrinKKN15, DBLP:conf/stacs/AustrinKKN15} in fact solves the more general \textsc{Linear Sat} problem. In \textsc{Linear Sat} one is given an integer $t$, matrix $A \in \mathbb{Z}^{n\times m}_{2}$, and vectors $\vec{b} \in \mathbb{Z}^n_2$ and $\vec{\omega} \in \mathbb{N}^m$. The task is to find $\vec{x} \in \mathbb{Z}^m_2$ satisfying $A\vec{x}\equiv\vec{b}$ and $\vec{\omega} \cdot \vec{x} \leq t$.

\begin{theorem}\label{thm:linsat}
There is an $O^*(2^{0.3399m})$-time Monte Carlo algorithm solving \textsc{Linear Sat}.
\end{theorem}

To our best knowledge no $O^*(2^{(0.5-\Omega(1))m})$-time algorithm for \textsc{Linear Sat} was known before. We get as a corollary that, given a bipartite graph $G=(F\dot{\cup}U,E)$, we can determine the smallest size of a set partition in time $O^*(2^{0.3399m})$. We take this as a first signal that the representation method is useful for solving \setpartition (and \setcover) for instances with small universe. To see this consequence, note we can reduce this problem to \textsc{Linear Sat} as follows: For every $f \in F$ add the incidence vector of $N(f)$ as a column to $A$, and set the cost $\omega_i$ of picking this column to be $n|N(f)|+1$. Then the minimum of $\vec{\omega} \cdot \vec{x}$ subject to $A\vec{x}\equiv\vec{1}$ will be $n^2+s$ where $s$ is the number of sets in a minimum set partition. Let us remark that \cite[Page 130]{Dahllof22042} solves (a counting version) of \setpartition in time $O^*(1.2561^m)=O^*(2^{0.329m})$, and Drori and Peleg~\cite{Drori2002473} solve the problem in $O^*(2^{0.3212m})$ time,\footnote{We attempted to find any more recent faster algorithm, but did not find this. Though, we would not be surprised if using more recent tools in branching algorithms as~\cite{DBLP:journals/jacm/FominGK09} one should be able to more significantly outperform our algorithm for \setpartition.} so by no means our algorithm is the fastest in this setting. However, both use sophisticated branching and we find it intriguing that the representation method does work quite well even for the seemingly more general \textsc{Linear Sat} problem.

\section{Exploiting the Presence of Many Witness $\beta$-halves.}\label{sec:manyhalves}
For convenience we will work with \setpartition in this section; the results straightforwardly extend to \setcover but we will not need this in the subsequent section.

\begin{definition}
Given an $(n,m,s)$ instance of \setpartition, a subset $W \subseteq U$ is said to be a \emph{witness $\beta$-halve} if $|W|\in (\tfrac{1}{2}\pm\beta)n$ and there exist disjoint subsets $S_1,S_2 \subseteq F$ such that $N(S_1 \cup S_2)=U$, $\sum_{f \in S_1 \cup S_2}|N(f)|=n$, $N(S_1) = W$, $N(S_2) = U \setminus W$ and $|S_1|+|S_2|=s$.
\end{definition}

Note that this is similar to the intuitive definition outlined in Section~\ref{sec:intro}, except that we require $|W|\in (\tfrac{1}{2}\pm\beta) |U|$ and we adjusted the definition to the \setpartition problem. Since $S_1 \cup S_2$ is a set partition of size $s$ we see that if a witness $\beta$-halve exists, we automatically have a yes instance.

In this section we will give randomized algorithms that solve promise-variants of \setpartition with the promise that, if the instance is a yes-instance, there will be an exponential number of witness halves that are sufficiently balanced (i.e. of size close to $n/2$). In the first subsection we outline the basic algorithm and in the second subsection we show how tools from the literature can be combined with our approach to also give a faster algorithm if the number of sets is exponential in $n$.

\subsection{The basic algorithm}
\begin{theorem}\label{thm:main}
There exists an algorithm $\mathtt{A1}$ that takes an $(n,m,\solsize)$-instance of \setpartition and real numbers $\beta,\zeta >0$ satisfying $2\sqrt{\beta} \leq \zeta < 1/4$ as input, runs in time $2^{(1-(\zeta/2)^4)n}\poly(n)m$, and has the following property: If there exist at least $\Omega(2^{\zeta n})$ witness $\beta$-halves it returns \yes with at least constant probability, and if there does not exist a set partition of size $\solsize$ it returns \no.
\end{theorem}
Note that the theorem does not guarantee anything on Algorithm~$\mathtt{A1}$ if a partition of $\solsize$ sets exists and there are only few witness halves, but we will address this later. A high level description of the Algorithm~$\mathtt{A1}$ is given in Figure~\ref{alg:mainalg}:
\begin{figure}[H]
\begin{framed}
\begin{algorithmic}[1]
\REQUIRE $\mathsf{A1}(G=(F\mathop{\dot{\cup}}U,E), \solsize,\zeta,\beta)$. \hfill\algorithmiccommentt{Assumes $2\sqrt{\beta}\leq \zeta < 1/4$}
\ENSURE An estimate of whether there exists a set partition of size $\solsize$.
\FOR{integer $l$ satisfying $\lfloor(1/2-\beta)n\rfloor < l < \lceil(1/2+\beta)n\rceil$}\label{line:startloop}
	\STATE Sample $\cW \subseteq \binom{U}{l}$ by including every set of $\binom{U}{l}$ with probability $2^{-\zeta n}$.\label{lin:smpl}
	\STATE For every $W\in \cW$ and $i \in [n]$, compute $c_i(W)$ and $c_i(U \setminus W)$.\label{lin:lem}
	\LineIf{$\exists i \in [n]: c_i(W) \wedge c_{s-i}(U \setminus W) $}{\algorithmicreturn\ \yes.}
\ENDFOR
\STATE \algorithmicreturn\ \no.
\end{algorithmic}
\end{framed}
\caption{High level description of the Algorithm implementing Theorem~\ref{thm:main}.}
\label{alg:mainalg}
\end{figure}

Here, we define $c_i(W)$ to be true if and only if there exists $S_1 \subseteq F$ with $|S_1|=i$, $N(S_1)=W$, and for every $f,f' \in S_1$ with $f \neq f'$, $N(f) \cap N(f') = \emptyset$.
Given a set family $\cW$, we denote ${\downarrow}\cW=\{X: \exists W \in \cW \wedge X \subseteq W \}$ for the \emph{down-closure} of $\cW$. The following lemma concerns the sub-routine invoked in Algorithm~\ref{alg:mainalg} and can be proved via known dynamic programming techniques, and is postponed to Appendix~\ref{app:knowndp}.

\begin{lemma}\label{lem:downclosure1}
There exists an algorithm that given a bipartite graph $G=(F \mathop{\dot{\cup}} U,E)$ and $\cW \subseteq 2^{U}$ with $|U|=n$ and $|F|=m$, computes $c_i(W)$ for all $W \in \cW$ and $i \in [n]$ in $O(\poly(n)|{\downarrow}\cW| m)$ time.
\end{lemma}

Thus, for further preparation of the proof of Theorem~\ref{thm:main}, we need to analyze the maximum size of the (down/up)-closure of $\cW$ in Algorithm~$\mathtt{A1}$ in Figure~\ref{alg:mainalg}.

\begin{lemma}\label{lem:bndclosure}
Let $\zeta>0$, $\beta$ (which may be negative) be real numbers satisfying $2\sqrt{|\beta|} \leq \zeta < 1/4$ and $|U|=n$. Suppose $\cW\subseteq \binom{U}{(1/2 +\beta)n}$ with $|\cW|\leq 2^{(1-\zeta)n}$. Then $|{\downarrow}\cW| \leq n2^{(1-(\zeta/2)^4)n}$.
\end{lemma}
\begin{proof}
Let $\lambda \leq \beta$ and $w_\lambda = |\{ W \in {\downarrow}\cW: |W|=\lambda n\}|$, so $|{\downarrow}\cW| \leq n\cdot\max_{\lambda}w_\lambda$. Then we have the following upper bounds:
\[
	w_\lambda \leq \binom{n}{\lambda n} \leq 2^{h(\lambda)n}, \qquad w_\lambda \leq |\cW|\binom{(1/2+\beta) n}{\lambda n} \leq 2^{\left((1-\zeta)+h\left(\frac{\lambda}{1/2+\beta}\right)(1/2+\beta)\right) n}.
\]
To see the second upper bound, note that any set $W \in \cW$ can have at most $\binom{(1/2+\beta) n}{\lambda n}$ subsets of size $\lambda n$. Thus, we see that $|{\downarrow}\cW|/n$ is upper bounded by $2^{f(\zeta,\beta)n}$, where
\[
	f(\zeta,\beta) = \max_{\lambda\leq 1/2+ \beta} \min \left\{h(\lambda), (1-\zeta)+h\left(\frac{\lambda}{1/2+\beta}\right)(1/2+\beta) \right\}.
\]
The remainder of the proof is therefore devoted to upper bounding $f(\zeta,\beta)$. We establish this by evaluating both terms of the minimum, setting $\lambda$ to be $\lambda'=(1-\zeta^2)(1/2+\beta)$. 
First note that by our assumption
\begin{align*}
	\lambda' = (1-\zeta^2)(1/2+\beta) &= 1/2-\zeta^2/2+\beta-\zeta^2\beta \leq 1/2-\zeta^2/2+\zeta^2/4-\zeta^2\beta &< 1/2,\\
	\lambda' / (1/2+\beta) &= 1-\zeta^2 &> 1/2.
\end{align*}

Therefore, since $h(x)$ is increasing for $x < 1/2$, $h(\lambda) \leq h(\lambda')$ for $\lambda \leq \lambda'$. Similarly, $h\left(\frac{\lambda}{1/2+\beta}\right)$ is at most $h\left(\frac{\lambda'}{1/2+\beta}\right)$ for $\lambda \geq \lambda'$, and we may upper bound $f(\zeta,\beta)$ by the maximum of the two terms of the minimum in $f(\zeta,\beta)$ obtained by setting $\lambda=\lambda'$. For the first term of the minimum, note that by Lemma~\ref{fact:bounds}, Item 1:
\begin{align*}
	h(\lambda') = h((1-\zeta^2)(1/2+\beta)) &\leq 1-(1/2-(1-\zeta^2)(1/2+\beta))^2\\
																					&= 1-\left( \zeta^2/2-\beta+\beta\zeta^2 \right)^2\\
																					&\leq 1-\left( \zeta^2/2-\beta \right)^2\\
																					&\leq 1-(\zeta^2/4)^2 = 1-(\zeta/2)^4.
\end{align*}
For the second term we have
\begin{align*}
 1-\zeta + h\left(\frac{\lambda'}{1/2+\beta}\right)(1/2+\beta) &=	1-\zeta + h(1-\zeta^2)(1/2+\beta) &\hfill\algorithmiccommentt{by Lemma~\ref{fact:bounds}, Item 2}\\
																							 &= 1-\zeta + h(\zeta^2)(1/2+\beta) &\algorithmiccommentt{$\beta < \frac{1}{64}$ by assumption} \\
																							 &\leq 1-\zeta+\zeta^2\lg\left(\frac{4}{\zeta^2}\right)\tfrac{33}{64}&\algorithmiccommentt{$\zeta\lg\left(\frac{4}{\zeta^2}\right) \leq \tfrac{3}{2}$}\\
																							 &\leq 1-\zeta+\zeta\tfrac{3}{2}\cdot\tfrac{33}{64}&\\
																							 &\leq 1- \zeta/10.&
\end{align*}
note for the penultimate inequality that $\zeta\lg(\frac{4}{\zeta^2})$ is monotone increasing for $0\leq \zeta \leq 1/4$ and substituting $\zeta=1/4$ in this expression thus upper bounds it with $3/2$.
\end{proof}

Now we are ready to wrap up this section with the proof of Theorem~\ref{thm:main}:
\begin{proof}[Proof of Theorem~\ref{thm:main}]

We can implement Line~\ref{lin:lem} by invoking the algorithm of Lemma~\ref{lem:downclosure1} with both $|\cW|$ and $\cW' = \{W: [n] \setminus W \in \cW\}$. This will take time $O(\poly(n)(|{\downarrow}W|+|{\downarrow} \cW'|)m)$. This is clearly the bottleneck of the algorithm, so it remains to upper bound (the expectation of) $|{\downarrow}W|+|{\downarrow} \cW'|$ by applying Lemma~\ref{lem:bndclosure}. To do this, note that $\cW \subseteq \binom{n}{l}$, $\cW' \subseteq \binom{n}{n-l}$, and we have that $(1/2-\beta)n \leq l,n-l \leq (1/2+\beta)n$. Also, $2\sqrt{|\beta|} \leq \zeta$ by assumption so indeed Lemma~\ref{lem:bndclosure} applies. Then on expectation $|\cW| \leq \binom{n}{(1/2+\beta) n}2^{-\zeta n} \leq 2^{(1-\zeta)n}$, and thus the running time\footnote{Due to the sampling in Line~\ref{lin:smpl}, we actually only get an upper bound on the expectation of the running time, but by Markov's inequality we can simply ignore iterations where $\cW$ exceeds twice the expectation.} indeed is as claimed. 

For the correctness, it is easily checked that the algorithm never returns false positives. Moreover, if there exist at least $\Omega(2^{\zeta n})$ witness $\beta$-halves then for some $l$ in the loop of Line~\ref{line:startloop}, there are at least $\Omega(2^{\zeta n}/n)$ witness halves of size $l$. Thus in this iteration we see by Lemma~\ref{fact:bounds}, Part 3 that
\begin{equation}\label{eq:probwit}
	\Pr[\nexists \text{ witness halve } W \in \cW ] \leq \left(1-\frac{1}{2^{\zeta n}}\right)^{\Omega(2^{\zeta n}/n)} \leq e^{-1/n}.
\end{equation}
and if a witness halve $W \in \cW$ exists the algorithm returns $\mathbf{yes}$ since $c_i(W) \wedge c_{s-i}(U \setminus W)$ holds for some $i$ by the definition of witness halve.
Therefore, if we perform $n$ independent trials of Algorithm $\mathtt{A_1}$ it return $\mathbf{yes}$ with probability at least $1-1/e$.
\end{proof}

\subsection{Improvement in the case with exponentially many input sets.}

In this section we show that under some mild conditions, the existence of many witness halves can also be exploited in the presence of exponentially many sets. This largely builds upon machinery developed by Bj\"orklund et al.~\cite{DBLP:journals/siamcomp/BjorklundHK09,DBLP:journals/mst/BjorklundHKK10}. To state our result as general as possible we assume the sets are given via an oracle so our running can be sublinear in the input if the number of sets is close to $2^n$.

\begin{theorem}\label{thm:manysets}
There exists an algorithm that, given oracle access to an $(n,m,\solsize)$-instance of \setpartition and real numbers $\beta,\zeta >0$ satisfying $2\sqrt{\beta} \leq \zeta < 1/4$, runs in time $2^{(1-(\zeta/2)^4)n}\poly(n)T$ and has the following property: if there exist at least $\Omega(2^{\zeta n})$ witness $\beta$-halves, it outputs \yes with constant probability and if there does not exist a set partition of size $\solsize$ it outputs \no.

Here the oracle algorithm accepts $X\subseteq U$ as input, and decides whether there exists $f \in F$ with $N(f)=X$ in time $T$.
\end{theorem}

The proof of Theorem~\ref{thm:manysets} is identical to the proof of Theorem~\ref{thm:main} (and therefore omitted), except that here we use the following lemma instead of Lemma~\ref{lem:downclosure1}:

\begin{lemma}\label{lem:downclosure2}
There exists an algorithm that, given $\cW \subseteq 2^{U}$ and oracle access to a bipartite graph $G=(F \mathop{\dot{\cup}} U,E)$, computes the values $c_i(W)$ for all $W \in \cW$ in $O(T|{\downarrow}\cW|\poly(n))$ time. Here the oracle algorithm accepts $X\subseteq U$ as input, and decides whether there exists $f \in F$ with $N(f)=X$ in time $T$.
\end{lemma}
This lemma mainly reiterates previous work developed by Bj\"orklund et al.~\cite{DBLP:journals/siamcomp/BjorklundHK09,DBLP:journals/mst/BjorklundHKK10}, but since they did not prove this lemma as such we include a proof here in Appendix~\ref{subsec:yates}.

\section{Finding Large Set Covers Faster}\label{sec:mainproofs}
In this section we will use the tools of the previous sections to prove our main results, Theorems~\ref{thm:largesol} and~\ref{thm:col}. We first connect the existence of large solutions to the existence of many witness halves in the following lemma:

\begin{lemma}\label{lem:largesolutionmanywitnesssplits}
If an $(n,m,\solsize)$-instance of \setpartition has no empty sets and satisfies $\solsize \geq \sigma_0 n$ and $|N(f)| \leq \sigma_0^4n/8$ for every $f \in F$, there is a solution if and only if there exist at least $2^{\sigma_0 n}/4$ witness $(\sigma_0^2/4)$-halves.
\end{lemma}
\begin{proof}
Note that the backward direction is trivial since by definition the existence of a witness halve implies the existence of a solution.

For the other direction, suppose $S = \{f_1,\ldots,f_s\}$ is a set partition, and denote $d_i = |N(f_i)|$. Suppose $S' \subseteq S$ is obtained by including every element of $S$ with probability $1/2$ in $S'$. since $N(f_i)\cap N(f_j)=\emptyset$ for $i\neq j$, we have that the random variable $|N(S')|$ is a sum of $s$ independent random variables that equal $0$ and $d_i$ with probability $1/2$. By the Hoeffding bound (Lemma~\ref{lem:hoeff}) we see that
\[
	\Pr\left[\Big|\left| N(S') \right| - \mathbb{E}[|N(S')|] \Big| \geq n\sigma_0^2/4\right] \leq 2\cdot \exp\left( \frac{-n^2\sigma_0^4/8}{\sum_{e \in S}d_e^2}\right) \leq 2\cdot \exp\left( \frac{-n^2\sigma_0^4/8}{n^2\sigma_0^4/8}\right) < \tfrac{3}{4},
\]
where the second inequality follows from $d_e \leq \sigma_0^4n/8$ and $\sum_{e \in S}d_e =n$. So for at least $2^{|S|}/4 \geq 2^{\sigma_0 n}/4$ subsets $S' \subseteq S$ we have that $|N(S')| \in (\tfrac{1}{2}\pm \sigma_0^2/4)n$. Thus, since for each such $S'$, $N(S')$ determines $S'$ and thus gives rise to a distinct witness halve, there are at least $2^{\sigma_0 n}/4$ witness $(\sigma_0^2/4)$-halves.
\end{proof}

Now we are ready to prove the first main theorem, which we recall here for convenience.
\begin{reptheorem}{thm:largesol}
There is a Monte Carlo algorithm that takes an instance of \setcover on $n$ elements and $m$ sets and an integer $\solsize$ as input and determines whether there exists a set cover of size $s$ in $O(2^{(1-\Omega(\sigma^4))n}m)$ time, where $\sigma=s/n$.
\end{reptheorem}

\begin{proof}
The algorithm implementing Theorem~\ref{thm:largesol} is given in Figure~\ref{alg:alglargesol}.
\begin{figure}[H]
\begin{framed}
\begin{algorithmic}[1]
\REQUIRE $\mathsf{A2}(G=(F\mathop{\dot{\cup}}U,E),\sigma)$.
\STATE Ensure $|N(f)| \leq \sigma^4n/1000$ using Observation~\ref{obs:nolinear}\label{lin:redlin}.
\FOR{every integer $s$ satisfying $\lfloor \sigma n/2 \rfloor \leq s \leq \sigma n$}\label{lin:forl}
	\STATE Create an $(n,m',s)$-instance $((F' \mathop{\dot{\cup}} U,E),s)$ of \setpartition where $F'$ is constructed by adding a vertex $f'$ with $N(f')=X$ for all $ f \in F, X \subseteq N(f)$.\label{lin:setpart}
	\STATE Let $\sigma_0 = s/n$.
	\LineIf{$\mathsf{A1}((F' \mathop{\dot{\cup}} U,E),s,\sigma_0,\sigma_0^2/4)= \yes$}{\algorithmicreturn\ \yes.}\label{lin:subrout}
\ENDFOR
\STATE Pick an arbitrary subset $X \in \binom{U}{n/2}$.\label{lintr1}
\STATE Find the sizes $l$ and $r$ of the smallest set covers in the instances induced by elements $X$ and respectively elements $U \setminus X$ in $O(2^{n/2}\poly(n)m)$ time with standard dynamic programming.\label{lintr2}
\LineIfElse{$l+r <= \sigma n$}{\algorithmicreturn\ \yes}{\algorithmicreturn\ \no.}\label{lin:retyes2}
\end{algorithmic}
\end{framed}
\caption{Algorithm for \setcover large solutions (implementing Theorem~\ref{thm:largesol}).}
\label{alg:alglargesol}
\end{figure}

We first focus on the correctness of this algorithm. It is clear that the algorithm never returns false positives on Line~\ref{lin:subrout} since Algorithm~$\mathsf{A1}$ also has this property. If $\mathbf{yes}$ is returned on Line~\ref{lin:retyes2} it is also clear there exists a solution.

Now suppose that a set cover $S$ of size at most $\sigma n$ exists. First suppose $\sigma n/2 \leq |S| \leq \sigma n$. We consider $\solsize=|S|$ in some iteration of the loop on Line~\ref{lin:forl}. Notice that now in Line~\ref{lin:setpart} we have reduced the problem to a yes-instance of \setpartition without empty sets satisfying $|N(f)| \leq \sigma^4n/1000$ for every $f \in F$. Therefore Lemma~\ref{lem:largesolutionmanywitnesssplits} applies with $\sigma_0\geq \sigma/2$ and we see there are at least $2^{\sigma_0 n}/4$ witness $(\sigma_0^2/4)$-halves. Thus, we can apply Theorem~\ref{thm:main} with $\zeta=\sigma_0$ and $\beta=\sigma^2_0/4$ to find the set $S$ with constant probability, since $\beta \leq (\zeta/2)^2$.

Now suppose $|S| \leq \sigma n/2$. Then picking every element in $S$ twice is a solution (as a multiset), and it implies that for every $X \subseteq U$ the sizes of the smallest set covers $l$ and $r$ (as defined in the algorithm) satisfy $l+r \leq \sigma n$. Thus Lines~\ref{lintr1}-\ref{lin:retyes2} find such a set and the algorithm returns $\mathbf{yes}$.

For the running time, Line~\ref{lin:redlin} takes at most $O(2^{(1-\sigma^4/1000)n}\poly(n)m)$ due to Observation~\ref{obs:nolinear}. For Line~\ref{lin:subrout}, due to Theorem~\ref{thm:main} this runs in time 
\begin{align*}
	O(2^{(1-(\zeta/2)^4)n}\poly(n)m') &= O(2^{(1-(\sigma_0/2)^4)n}\poly(n)2^{\sigma^4n/1000}m) \\
	&\leq O(2^{(1-(\sigma/4)^4)n}\poly(n)2^{\sigma^4n/1000}m) \\
	&= O(2^{(1+\sigma^4(1/1000-1/4^4))n}\poly(n)m) \\
	&\leq O(2^{(1-\Omega(\sigma^4))n}\poly(n)m).
\end{align*}
as claimed in the theorem statement.
\end{proof}

As a more direct consequence of the tools of the previous section we also get the following result for \setpartition:

\begin{theorem}\label{thm:manysets2}
There exists a Monte Carlo algorithm for \setpartition that, given oracle access to an $(n,m,\sigma n)$-instance satisfying $0 < |N(f)| \leq \sigma^4 n/8$ for every $f \in F$, runs in $2^{(1-\Omega(\sigma^4))n}\poly(n)T$ time.

Here the oracle algorithm accepts $X\subseteq U$ as input, and decides whether there exists $f \in F$ with $N(f)=X$ in time $T$.
\end{theorem}
\begin{proof}
Lemma~\ref{lem:largesolutionmanywitnesssplits} implies the instance is a YES-instance if and only if there exist  $2^{\sigma n}/4$ witness $(\sigma^2/4)$-halves. Thus Theorem~\ref{thm:manysets} implies the theorem statement.
\end{proof}

Note that this theorem also implies an $O((m+2^{(1-\Omega(\sigma^4))n})\poly(n))$ time algorithm for \setpartition where the sets are given explicitly because we can construct a binary search tree after which we can implement the oracle to run in $T=n$ query time. We remark that it would be interesting to see whether the assumption $|N(f)| \leq \sigma^2 n/4$ is needed, but removing this assumption seems to require more ideas than the ones from this work: For example if the solution has three sets of size $3n/10$ there will be no witness halve that is sufficiently balanced, and alternatively using Observation~\ref{obs:nolinear} seems to be too slow.

However, if we settle for a additive $1$-approximation we can deal with this issue in a simple way and have as a particular consequence the second result mentioned in the beginning of this paper:
\begin{reptheorem}{thm:col}
There is a randomized algorithm that given graph $G$ and integer $\solsize=\sigma n$, in $O^*(2^{(1-\Omega(\sigma^4))n})$ time outputs
\yes with constant probability, if $\chi(G) < \solsize$, and \no, if $\chi(G) > \solsize$.
\end{reptheorem}
\begin{proof}
Let $G=(V,E)$ and define a \setpartition instance where for every independent set $I \subseteq V$ of $G$ there is an element $f \in F$ with $N(f)=I$. It is easy to see that this instance of \setpartition has a solution of size $s$ if and only if $\chi(G) \leq s$. 

Check in $\binom{n}{\sigma^4 n/8}$ time whether $G$ has an independent set of size $\sigma^4n/8$. If such an independent set is found, remove this set from the graph and return yes if the obtained graph has a $(k-1)$-coloring and no otherwise. Using the $O^*(2^{n})$ time algorithm by Bj\"orklund et al.~\cite{DBLP:journals/siamcomp/BjorklundHK09} in the second step, this procedure clearly runs in time $O^*(2^{(1-\Omega(\sigma^4))n})$, and always finds a coloring using at most one more color than the minimum number of colors if a large enough independent set exists.

On the other hand, if the maximum independent set of $G$ is of size at most $\sigma^4n/8$, we may apply Theorem~\ref{thm:manysets2} with $T=\poly(n)$ since it can be verified in polynomial time whether a given $X \subseteq V$ is an independent set, and the theorem follows.
\end{proof}

\section{Directions for Further Research}\label{sec:concremark}
In this section, we relate the work presented to some notorious open problems. The obvious open question is to determine the exact complexity of the \setcover problem:

\begin{openproblem}
Can \setcover be solved in time $O^*((2-\Omega(1))^n)$?
\end{openproblem}
This question was already stated at several places. It is known that if a version of \setcover where the number of solutions modulo $2$ is counted can be solved in $(2-\Omega(1))^{n}$ the Strong Exponential Time Hypothesis fails. We refer to~\cite{DBLP:conf/coco/CyganDLMNOPSW12}, for more details.

Less ambitiously, it is natural to wonder whether our dependency on $\sigma$ can be improved. Our algorithm and analysis seem loose, but we feel the gain of a sharpening this analysis does not outweigh the technical effort currently: For a better dependence, we need both a better bound in Lemma~\ref{lem:bndclosure} and to reduce the set sizes more efficiently than in Observation~\ref{obs:nolinear}. As further research we suggest to find a different algorithmic way to deal with the case where many witness halves are unbalanced. But this alone will not suffice to give linear dependence in $\sigma$  since in Lemma~\ref{lem:bndclosure} we do not expect to get linear dependence on $\zeta$ even if $\beta=0$. It would also be interesting to see which other instances of \setcover can be solved in $O^*((2-\Omega(1))^n)$ time. One that might be worthwhile studying is whether this includes instances with optimal set covers in which the sum of the set sizes is at least $(1+\Omega(1))n$; one may hope to find exponentially many (balanced) witness halves here as well.

In~\cite{DBLP:conf/coco/CyganDLMNOPSW12}, the authors also give a reduction from \subsetsum{} to \setpartition{}. The exact complexity of \subsetsum{} with small integers is also something we explicitly like to state as open problem here, especially since the $O^*(t)$ time algorithm (where $t$ is the target integer) is perhaps one of the most famous exponential time algorithms:
\begin{openproblem}
Can \subsetsum{} with target $t$ be solved in time $O^*(t^{1-\Omega(1)})$, or can we exclude the existence of such an assuming the Strong Exponential Time Hypothesis?
\end{openproblem}
Note this question was before asked in~\cite{husfeldt_et_al:DR:2013:4342} by the present author. It would be interesting to study the complexity of \subsetsum{} in a similar vein as we did in this paper: are there some special properties allowing a faster algorithm? For example, a faster algorithm for instances of high `density' (e.g., $n/\lg t$) may be used for improving an algorithm of Horowitz\&Sahni~\cite{HorowitzSahni74}. Note that here the `density' of a \subsetsum{} instance is the inverse of what one would expect when relating to the definition of density of $k$-CNF formula.

Another question that has already open for a while is

\begin{openproblem}
Can \textsc{Graph Coloring} be solved in time $O^*(2^{(1-\Omega(1))n})$?
\end{openproblem}
Could the techniques of this paper be used to make progress towards resolving this question? While our algorithm seems to benefit from the existence of many optimal colorings, in an interesting paper Bj\"orklund~\cite{DBLP:journals/corr/Bjorklund15} actually shows that the existence of \emph{few} optimal colorings can be exploited in graphs of small pathwidth. Related to this is also the Hamiltonicity problem. In our current understanding this problem becomes easier when there is a promise that there are few Hamiltonian cycles (see~\cite{DBLP:conf/icalp/BjorklundDH15}, but also e.g.~\cite{DBLP:journals/siamcomp/ChariRS95} allows derandomizations of known probabilistic algorithms in this case), so a natural approach would be to deal explicitly with instances with many solutions by sampling dynamic programming table in a similar vein as done in this paper.

\subparagraph{Acknowledgements}
The author would like to thank Per Austrin, Petteri Kaski, Mikko Koivisto for their collaborations resulting in~\cite{DBLP:conf/stacs/AustrinKKN15,DBLP:journals/corr/AustrinKKN15} that mostly inspired this work, Karl Bringmann for discussions on applications of~\cite{DBLP:journals/corr/AustrinKKN15} to \textsc{Set Partition}, and anonymous reviewers for their useful comments.

\bibliographystyle{plain}
\bibliography{setcover}

\appendix

\makeatletter
\edef\thetheorem{\expandafter\noexpand\thesection\@thmcountersep\@thmcounter{theorem}}
\makeatother

\section{Missing Proofs}\label{app:A}

\subsection{Proof of Theorem~\ref{thm:redcovpart}}\label{app:redcovpart}
Given an $(n,m,s)$-instance of \setcover, use Observation~\ref{obs:nolinear} to assure that for every $f \in F$, $|N(f)| \leq \epsilon n$. It is easy to see that this instance is a YES-instance if and only if the \setpartition instance $(G',s)$ is a YES-instance, where $G'$ is the bipartite graph with $U$ on one side and $F'$ on the other with for every $f \in F$ and $X \subseteq N(f)$ a vertex in $F'$ with neighborhood $X$.

\subsection{Proof of Theorem~\ref{thm:covparteq}}\label{app:covparteq}
For the backward direction, an $O(m^c2^{(1-\epsilon')n})$ algorithm for \setpartition concatenated to the reduction of Theorem~\ref{thm:redcovpart} with $\epsilon=\epsilon'/(2c)$ gives a 
\[
	O\left(\left(m2^{\epsilon' n/(2c)}\right)^c 2^{(1-\epsilon')n}+m2^{(1-\epsilon)n}\right) \leq O(m^c2^{(1-\epsilon'/2)n}+m2^{(1-\epsilon)n})
\]
time algorithm for \setcover.
	
For the forward direction, given an $(n,m,s)$-instance of \setpartition first use Observation~\ref{obs:redsolsize} with $c = 2/ \epsilon$ to obtain an $(n,m^{2/\epsilon},s\epsilon/2)$ instance of \setpartition (without loss of generality, we may assume $s$ divides $2/\epsilon$ by decreasing $\epsilon$ and adding at most $2/\epsilon$ sets and elements). Denoting $s'=s\epsilon/2$, now iterate over all unordered integer partitions $a_1+\ldots+a_{s'} = n$, and for each such partition solve a \setcover instance $(((F' \dot{\cup} U'),E'),s')$ constructed as follows from the \setpartition instance with the assumed algorithm:
 \begin{itemize}
		\item For every $i=1,\ldots,s'$ add an element $e_i$ to $U'$
		\item For every $f \in F$ and $i$ such that $|N(f)|=a_i$ add a set $f'$ to $F'$ with $N(f')=f \cup \{e_i\}$.
 \end{itemize}
Note that possibly $a_i=a_j$ for $i \neq j$ we make at least two copies of every set $f \in F$ with $|N(f)|=a_i$. It is easy to see that this new \setcover instance is a YES-instance if and only if in the original \setpartition instance there is a set partition of size $s'$: each set only contains one element from $\{e_1,\ldots,e_{s'}\}$ so a Set Cover of size $s'$ needs to correspond with $f_1,\ldots,f_{s'} \in F$ satisfying $\sum_{i=1}^{s'}|N(f_i)|=n$.

 By a result of Hardy and Ramanujan~\cite{Hardy01011918}, there are at most $2^{O(\sqrt{n})}$ unordered integer partitions and hence concatenating this reduction with the assumed $O(m^{c}2^{(1-\epsilon)n})$ time \setcover algorithm leads to an 
\[
	O(2^{O(\sqrt{n})}m^{c2/\epsilon}2^{(1-\epsilon)(n+s')}) \leq O^*(2^{(1-\epsilon)(1+\epsilon/2)n})\leq O^*(2^{(1-\epsilon/2)n})
\]
time algorithm for \setpartition.

\subsection{Proof of Theorem~\ref{thm:fewsets}}\label{app:fewsets}
\begin{theorem}[\cite{DBLP:conf/iwpec/Koivisto09}]\label{thm:koiv}
Given an instance $n$-element set $N$, an integer $r$ and a family $\mathcal{F}$ of subsets of $N$ each of cardinality at most $r$, the partitions of $N$ into a give number of members of $\mathcal{F}$ can be counted in time $|\mathcal{F}|2^{\lambda_r n}\mathsf{poly}(n)$, where $\lambda_r = (2r-2)/\sqrt{(2r-1)^2-2\ln 2}$.
\end{theorem}

Note that the above result also immediately implies an algorithm deciding \setcover in time $2^r|\mathcal{F}|2^{\lambda_r n}\mathsf{poly}(n)$ since we may reduce this \setcover instance to a \setpartition instance by adding all subsets of the given sets. It is easily seen that for $r \geq 3$, $\lambda_r$ is sandwiched as 
\[
	1/2 \leq \lambda_r \leq (2r-2)/\sqrt{(2r-1.5)^2} = 1- \frac{1}{O(r)}.
\]

The algorithm implementing the theorem is given in Figure~\ref{alg:algfewsets}.

\begin{figure}[H]
\begin{framed}
\begin{algorithmic}[1]
\REQUIRE $\mathsf{A3}(G=(F\dot{\cup}U,E),s,r)$.
\IF{$\exists f \in F: |N(f)| \geq r$}
	\RETURN $\mathsf{A3}(G[F\setminus\{f\} \cup U],s,r) \vee \mathsf{A2}(G[F\setminus\{f\} \cup U\setminus N(f)],s-1,r)$.\label{line:branch}
\ELSE
	\STATE Use the algorithm of Theorem~\ref{thm:koiv}.\label{line:koiv}
\ENDIF
\end{algorithmic}
\end{framed}
\caption{Algorithm for \setcover or \setpartition with few sets (implementing Theorem~\ref{thm:fewsets}).}
\label{alg:algfewsets}
\end{figure}

We may bound the running time of Algorithm~$\mathsf{A3}$ by analyzing the branching tree of recursive calls where an execution reaching Line~\ref{line:koiv} represents a leaf. Note that the depth of the branching tree is at most $m$ and if we refer to the second recursive call of Line~\ref{line:branch} as the right branch, on any path from the root to a leaf there are at most $n/r$ right branches since in each such recursion step we decrease the size of $U$ in one call with at least $r$. Thus the number of paths from the root to a leaf in this tree with $i$ right branches is at most $\binom{m}{i}$. For each such a branch the size of $U$ has become $n-ir$ and since Line~\ref{line:koiv} is reached all sets are of size at most $r$, Theorem~\ref{thm:fewsets} applies and thus we spend $m\cdot\mathsf{poly}(n)2^{\lambda_r(n-ir)}$ per leaf of this type.

Thus, denoting $\mu = m/n$, the total running time can be written as
\begin{align*}
	m\cdot\mathsf{poly}(n)\sum_{i=1}^{m}\binom{m}{i}2^{\lambda_r(n-ir)} &= m\cdot\mathsf{poly}(n)2^{\lambda_r n} \left(1+\frac{1}{2^{\lambda_rr}}\right)^m\\
																										 &\leq m\cdot\mathsf{poly}(n)\left(2^{\lambda_r}\exp\left(\mu/2^{r/2}\right)\right)^n,
\end{align*}
where the equality follows from the binomial theorem and the inequality uses $(1+\frac{1}{n})^n \leq e$. Denoting $\mu=m/n$ and setting $r= \lceil 4\lg(\mu)\rceil$, we see that when taking asymptotics for growing $\mu$ (which is allowed since we may assume the running time complexity is monotone increasing with $\mu$), the running time becomes 
\[
 m\cdot\mathsf{poly}(n)\left(2^{1-\frac{1}{O(\lg \mu)} + \lg(e)/\mu}\right)^n = m\cdot\mathsf{poly}(n)2^{n-\frac{n}{O(\lg(m/n))}}.
\]

\subsection{Proof of Theorem~\ref{thm:subtle}}\label{app:subtle}
Given an arbitrary instance of \setpartition consisting of $G=(F \mathop{\dot{\cup}} U, E)$ and integer $s$, first construct an equivalent $(n,m',s')$ instance using Observation~\ref{obs:redsolsize} with $s'\leq\min\{\epsilon_1,\epsilon_2/2\}n$. Then add vertices $e_1,\ldots,e_{s'}$ to $U$ and replace every $f \in F$ with $s'$ copies $f_1,\ldots,f_{s'}$ where $N(f_{i})=N(f) \cup e_i$. It is easy to see that in this \setpartition instance all set partitions are of size $s'$ and it has a set partition of size $s'$ if and only if the original instance has one of size $s$. Thus, to transform the new instance into an equivalent $(n+s',m',\epsilon_1(n+s'))$-instance, simply add sufficiently many isolated vertices to $F$. Running the assumed algorithm on this instance results thus in an $O^*(2^{(1-\epsilon_2)(n+s')})=O^*(2^{(1-\epsilon_2)(1+\epsilon_2/2)n})=O^*(2^{(1-\epsilon_2/2)n})$ time algorithm.

\subsection{Proof of Theorem~\ref{thm:linsat}}

Recall from the main text that the \textsc{Linear Sat} problem is defined as follows: given an integer $t$, matrix $A \in \mathbb{Z}^{n\times m}_{2}$ and vectors $\vec{b} \in \mathbb{Z}^n_2$ and $\vec{\omega} \in \mathbb{N}^m$ the task is to find $\vec{x} \in \mathbb{Z}^m_2$ satisfying $A\vec{x}\equiv\vec{b}$ and $\vec{\omega} \cdot \vec{x} \leq t$. 

We let $\vec{a_1},\ldots,\vec{a_m} \in \mathbb{Z}^n_2$ denote the column vectors of $A$, and denote $\wt(x)$ for the Hamming weight of a vector $\vec{x}\in\mathbb{Z}^m_2$.

The algorithm implementing the theorem is outlined in Algorithm~\ref{alg:linsat}. The algorithm assumes the Hamming weight of the vector $\vec{x}$ minimizing minimum $\vec{\omega x}$ is at most $2m/3$ (which is in order to facilitate the running time analysis). We can assume this since if the minimum set of columns summing to $\vec{b}$ consists of more than $2m/3$ columns, the rank of $A$ is at least $2m/3$ (as the solution needs to consist of linearly independent columns) and we can solve the problem in $O^*(2^{m/3})$ time by finding a column basis, iterating over all subsets of columns not in the basis and for each such compute the unique way to extend it to a set of columns summing to $\vec{b}$ if it exists (this is a standard technique called `Information Set Decoding' introduced in~\cite{MCeliece}). Note, to ensure the Hamming weight is at most $2m/3$, we cannot simply assume the solution is at most size at most $m/2$ by looking for the complement otherwise since this gives a maximization problem. 

\begin{figure}
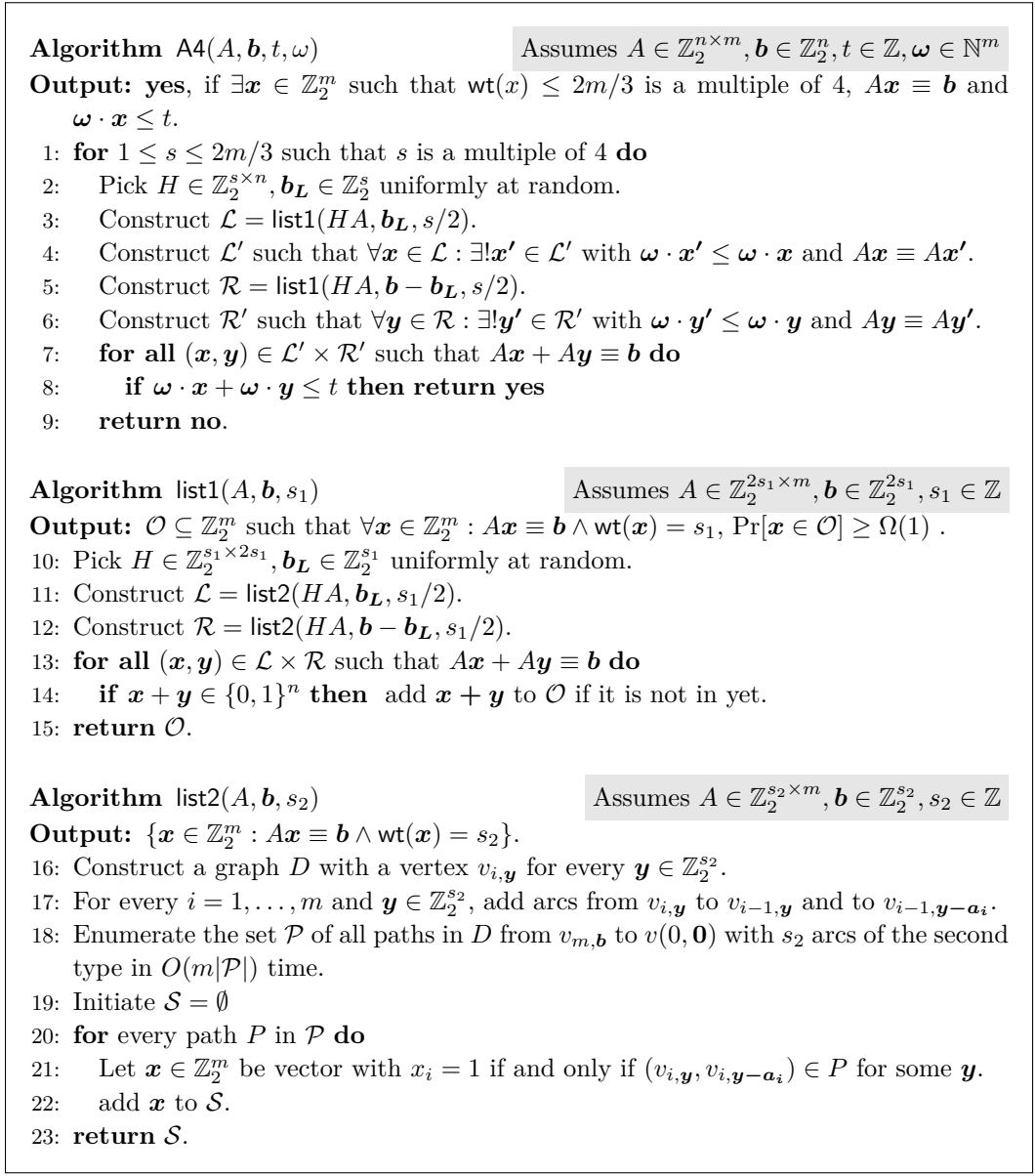

\begin{framed}
  \begin{algorithmic}[1]
    \REQUIRE $\mathsf{A4}(A,\vec{b},t,\omega)$\hfill\algorithmiccommentt{Assumes $A  \in \mathbb{Z}^{n \times m}_2, \vec{b} \in \mathbb{Z}^n_2, t \in \mathbb{Z},\vec{\omega}\in\mathbb{N}^m$}
    \ENSURE $\mathbf{yes}$, if $\exists\vec{x} \in \mathbb{Z}^m_2$ such that $\wt(x)\leq 2m/3$ is a multiple of $4$, $A\vec{x} \equiv \vec{b}$ and $\vec{\omega}\cdot\vec{x}\leq t$.
		\FOR{$1\leq s \leq 2m/3$ such that $s$ is a multiple of $4$}\label{lin:forloop}
		\STATE Pick $H \in \mathbb{Z}_2^{s \times n},\vec{b_L} \in \mathbb{Z}_2^{s}$ uniformly at random.
    \STATE Construct $\cL=\mathsf{list1}(HA,\vec{b_L},s/2)$. \label{alg1:build cL}
		\STATE Construct $\cL'$ such that $\forall \vec{x} \in \cL: \exists! \vec{x'} \in \cL'$ with $\vec{\omega}\cdot\vec{x'} \leq \vec{\omega}\cdot\vec{x}$ and $A\vec{x} \equiv A\vec{x'}$. \label{alg2:build cL2}
    \STATE Construct $\cR=\mathsf{list1}(HA,\vec{b}-\vec{b_L},s/2)$. \label{alg1:build cR}
		\STATE Construct $\cR'$ such that $\forall \vec{y} \in \cR: \exists! \vec{y'} \in \cR'$ with $\vec{\omega}\cdot\vec{y'} \leq \vec{\omega}\cdot\vec{y}$ and $A\vec{y} \equiv A\vec{y'}$. \label{alg2:build cR2}
    \FORALL{$(\vec{x},\vec{y}) \in \cL' \times \cR'$ such that $A\vec{x} + A\vec{y} \equiv \vec{b}$} \label{alg1:combine lists}
    \LineIf{$\vec{\omega}\cdot\vec{x}+\vec{\omega}\cdot\vec{y} \leq t$}{\algorithmicreturn\ $\mathbf{yes}$}\label{line:linyes}
    \ENDFOR
    \STATE \algorithmicreturn\ $\mathbf{no}$.
		\ENDFOR
  \end{algorithmic}
	\vspace{1em}
		\begin{algorithmic}[1]\setcounter{ALC@line}{9}
    \REQUIRE $\mathsf{list1}(A,\vec{b},s_1)$\hfill\algorithmiccommentt{Assumes $A  \in \mathbb{Z}^{2s_1 \times m}_2, \vec{b} \in \mathbb{Z}^{2s_1}_2, s_1 \in \mathbb{Z}$}
    \ENSURE $\mathcal{O} \subseteq \mathbb{Z}^m_2$ such that $\forall \vec{x}\in \mathbb{Z}^m_2:A\vec{x} \equiv \vec{b} \wedge \wt(\vec{x})=s_1$, $\Pr[\vec{x} \in \mathcal{O}] \geq \Omega(1)$ .
		\STATE Pick $H \in \mathbb{Z}_2^{s_1 \times 2s_1},\vec{b_L} \in \mathbb{Z}_2^{s_1}$ uniformly at random.
		\STATE Construct $\cL=\mathsf{list2}(HA,\vec{b_L},s_1/2)$.\label{lin:alg2call1}
		\STATE Construct $\cR=\mathsf{list2}(HA,\vec{b}-\vec{b_L},s_1/2)$.\label{lin:alg2call2}
		\FORALL{$(\vec{x},\vec{y}) \in \cL \times \cR$ such that $A\vec{x} + A\vec{y} \equiv \vec{b}$} \label{alg2:combine lists}
    \LineIf{$\vec{x}+\vec{y}\in\{0,1\}^n$}{ add $\vec{x+y}$ to $\mathcal{O}$ if it is not in yet.}
    \ENDFOR
		\STATE \algorithmicreturn\ $\mathcal{O}$.
  \end{algorithmic}

\vspace{1em}	
	\begin{algorithmic}[1]\setcounter{ALC@line}{15}
    \REQUIRE $\mathsf{list2}(A,\vec{b},s_2)$\hfill\algorithmiccommentt{Assumes $A  \in \mathbb{Z}^{s_2 \times m}_2, \vec{b} \in \mathbb{Z}^{s_2}_2, s_2 \in \mathbb{Z}$}
    \ENSURE $\{\vec{x} \in \mathbb{Z}^m_2: A\vec{x} \equiv \vec{b} \wedge \wt(\vec{x})=s_2\}$.
		\STATE Construct a graph $D$ with a vertex $v_{i,\vec{y}}$ for every $\vec{y} \in \mathbb{Z}^{s_2}_2$.\label{lin:constr1}
		\STATE For every $i=1,\ldots,m$ and $\vec{y} \in\mathbb{Z}^{s_2}_2$, add arcs from $v_{i,\vec{y}}$ to $v_{i-1,\vec{y}}$ and to $v_{i-1,\vec{y-a_i}}$.\label{lin:constr2}
		\STATE Enumerate the set $\mathcal{P}$ of all paths in $D$ from $v_{m,\vec{b}}$ to $v(0,\vec{0})$ with $s_2$ arcs of the second type in $O(m|\mathcal{P}|)$ time.\label{lin:enum}
		\STATE Initiate $\mathcal{S} = \emptyset$
		\FOR{every path $P$ in $\mathcal{P}$}
			\STATE Let $\vec{x}\in\mathbb{Z}_2^m$ be vector with $x_i=1$ if and only if $(v_{i,\vec{y}},v_{i,\vec{y-a_i}}) \in P$ for some $\vec{y}$.
			\STATE add $\vec{x}$ to $\mathcal{S}$.
		\ENDFOR
		\STATE \algorithmicreturn\ $\mathcal{S}$.
  \end{algorithmic}
	\end{framed}
	\caption{An $O^*(2^{0.3399n})$ time algorithm for \textsc{Linear Sat}.}
  \label{alg:linsat}
\end{figure}

\subparagraph{Running Time} First note that Line~\ref{lin:constr1} and Line~\ref{lin:constr2} can be implemented in time $2^{s_2}m$. Line~\ref{lin:enum} can be implemented with linear delay by elementary methods, and thus $\mathsf{list2}$ runs in $O(m(2^{s_2}+|\mathsf{list2}(A,\vec{b},s_2)|))$ time. For Algorithm~$\mathsf{list1}$, Line~\ref{lin:alg2call1} and Line~\ref{lin:alg2call2}, we see that $\mathbb{E}[|\mathsf{list2}(A,\vec{b},s_1/2)|]=\binom{m}{s_1/2}2^{-s_1}$, because any vector from $\mathbb{Z}^m_2$ of weight $s_1/2$ will be included with probability $2^{-s_1}$ in the output. Thus these lines run in expected time $O(m(2^{s_1/2}+\binom{m}{s_1/2}2^{-s_1}))$. 

Similarly, the for-loop at Line~\ref{alg2:combine lists} will take on expectation $|\mathcal{P}|2^{-s_1}$ iterations, where $\mathcal{P}$ is the set of pairs $(\vec{x},\vec{y})\in \left(\mathbb{Z}^m_2\right)^2$ such that $\wt(\vec{x})=\wt(\vec{y})=s_1/2$ and $A\vec{x}+A\vec{y}\equiv \vec{b}$. Here we divide by $2^{s_1}$ since this is the probability that such pair is in $\mathcal{L} \times \mathcal{R}$ because for this it needs to satisfy $HA\vec{x}\equiv \vec{b_L}$.

For algorithm $\mathsf{A4}$ at iteration $s$, Line~\ref{alg1:build cL} and Line~\ref{alg1:build cR} thus take expected time
\[
	O\left(m\left(2^{s_1/2}+\binom{m}{s_1/2}2^{-s_1}+\mathbb{E}[|\mathcal{P}|2^{-s_1}]\right)\right),
\]
where $s_1$ denotes $s/2$. Note that $\mathbb{E}[|\mathcal{P}|] \leq \binom{m}{s_1/2}^2/2^{s}$ since any pair $(\vec{x},\vec{y})$ in $\mathcal{P}$ also needs to satisfy $A\vec{x}+A\vec{y}\equiv \vec{b}$.  Denoting $\sigma=s/m$ and using $s_2=s_1/2=s/4$ we can thus upper bound the running time with
\[
 O\left(m2^{m\max_{\sigma \leq 2/3}\{\sigma/4,h(\sigma/4)-\sigma/2,2h(\sigma/4)-3\sigma/2\}}\right).
\]
It is easily verified by standard calculus or a Mathematica computation that the exponent in this expression is maximized for $\sigma=0.4444$ where it is at most $O(m2^{0.3399m})$.
Note that Line~\ref{alg2:build cL2} and Line~\ref{alg2:build cR2} can be implemented in time $O(m|\mathcal{L}|)$ and $O(m|\mathcal{R}|)$ with standard data structures and similarly we can efficiently enumerate all pairs in the for-loop at Line~\ref{alg1:combine lists}. Also note that in this for-loop the number of enumerated pairs is at most $|\mathcal{L}'|+|\mathcal{R}'|$ since all vectors $A\vec{x}$ for $\vec{x} \in \mathcal{L}'$ are different. Thus the algorithm indeed runs in the claimed running time.

\subparagraph{Correctness}
 For correctness, first note that whenever $\mathbf{yes}$ is returned on Line~\ref{line:linyes} this is clearly correct. Now suppose the instance of \textsc{Linear Sat} is a YES-instance. Then there exists a minimum weight $\vec{x} \in \mathbb{Z}^m_2$ satisfying $A\vec{x}\equiv\vec{b}$ and $\vec{\omega} \cdot \vec{x} \leq t$. Let us consider the iteration of the for-loop on Line~\ref{lin:forloop} where $s=\wt(\vec{x})$. Since $\vec{x}$ has minimum weight, the columns $i$ of $A$ for which $x_i=1$ need to be linearly independent since otherwise leaving out a linear combination would result in a smaller weight vector $\vec{x}$. This means that if we denote $\cZ=\{ A\vec{y}: \vec{y} \subseteq \vec{x} \wedge \wt(\vec{y})=s/2 \}$, then $|\cZ| = \binom{s}{s/2}$, where $\vec{y} \subseteq \vec{x}$ denotes $y_i \leq x_i$ for every coordinate $i$.

For $\vec{v}\in\mathbb{Z}^s_2$, let $f(\vec{v}) = |\{\vec{z} \in \cZ:H\vec{z}=\vec{v}\}|$ be an indicator function. We see that
\[
	\mathbb{E}\left[\sum_{\vec{v}\in\mathbb{Z}^s_2}f(\vec{v})^2\right] = \sum_{\vec{z_1},\vec{z_2} \in \cZ^2}\Pr[H(\vec{z_1}-\vec{z_2})=0] \leq |\cZ|+|\cZ|^2 2^{-s} \leq 2|\cZ|. 
\]
By Markov's inequality we thus see that $\Pr[\sum_{x\in\mathbb{Z}^s_2}f(x)^2 \leq 4|\cZ|] \geq 1/2$ over the choice of $H$. Conditioned on this, the Cauchy-Schwarz inequality implies that the number of $\vec{x} \in \mathbb{Z}^s_2$ such that $f(\vec{x})>0$ is at least $|\cZ|^2/(2|\cZ|) =|\cZ|/2$. When this happens, we have with probability at least $\binom{s}{s/2}/2^s=\Omega(1/\sqrt{s})$ that $\vec{b_L} \in \cZ$. If $\vec{b_L} \in \cZ$, we can apply the same reasoning to show that $\mathsf{list1}(HA,\vec{b_L},s/2)$ contains $\vec{y}$ with probability $\Omega(1/\sqrt{s})$ and $\mathsf{list1}(HA,\vec{b-b_L},s/2)$ contains $\vec{x}-\vec{y}$ with probability $\Omega(1/\sqrt{s})$, and the pair will be considered in the loop on Line~\ref{alg1:combine lists}. Since the latter two probabilities are independent we see that if a solution exists we find it with probability at least $\Omega(n^{-1.5})$, and thus we may repeat the algorithm $O(n^{1.5})$ times to ensure it finds a solution with constant probability if it exists.

\subsection{Proof of Lemma~\ref{lem:downclosure1}}
\label{app:knowndp}

Let $F=\{f_1,\ldots,f_m\}$ be arbitrarily ordered. For integers $i \in [n]$, $j \in [m]$ and $X \subseteq U$ define $c^i_j[X]$ to be true if and only if there exists $S_1 \subseteq \{f_1,\ldots,f_j\}$ such that $|S_1|=i$, $N(S_1)=X$ and for every $f,f' \in S_1$ with $f \neq f'$, $N(f) \cap N(f') = \emptyset$. We see that $c^i_0[X]$ is true if and only if $i=0$ and $X=\emptyset$, and for $j>0$ we have
\[
	c_j[X] = (c^{i-1}_{j-1}[X\setminus N(f_j)] \wedge N(f_j) \subseteq X) \vee c^{i}_{j-1}[X].
\]
The values $c_i(W)$ for $W \in \cW$ and $i \in [n]$ can be read off from $c^j_i(W)$, and because for computing $c^i_j[X]$ we only need the entries $c^i_j[Y]$ where $Y \subseteq X$, we can restrict our computation to computing $c^i_j(X)$ for $X \in {\downarrow}\cW$, and the runtime bound follows.

\subsection{Proof of Lemma~\ref{lem:downclosure2}}
\label{subsec:yates}

Define $f_x(X)$ to be true if and only if $\exists f \in F: N(f) = X$ and $|X|=x$. Note we can, within the claimed time bound, create a table storing the values $f_x(X)$ for all $X \in {\downarrow}\cW$ using the oracle. Now define $g_x(Y) = \sum_{Y \subseteq X}f_x(X)$. Let $U=\{u_1,\ldots,u_n\}$ and define
\[
	g^{j}_x(X) = \sum_{X \cap\{u_1,\ldots,u_j\} \subseteq Y \subseteq X}f_x(Y).
\]
Then we see that $g^{n}_x(X)=f_x(X)$, and for $j < n$ we can compute $g^j_{x}$ using
\begin{equation}\label{eq:yates}
	g^{j}_x(X) = [u_j \in X]g^{j+1}_x(X\setminus \{ u_j \})+g^{j+1}_x(X).
\end{equation}
Thus, by straightforward dynamic programming using~\eqref{eq:yates} we can compute $g_x(X)=g^0_x(X)$ for every $x$ and $X \in {\downarrow}\cW$ in $O(n^2|{\downarrow}\cW|)$ time. Next, define 
\[
	h_{x,i}(X)= \sum_{x_1+x_2+\ldots+x_i=x}g_{x_1}(X)g_{x_2}(X)\cdots g_{x_i}(X),
\]
or equivalently, $h_{x,i}$ is the number of tuples $f_1,\ldots,f_i$ such that $N(f_i) \subseteq X$ and $\sum_{l=1}^i |N(f_i)|=x$. Note that for a fixed $X$, given the entries $g_x(X)$ for every $x \leq n$, we can compute $h_{x,i}(X)$ in time $\poly(n)$ using standard dynamic programming or the Fast Fourier Transformation. Denote 
\[
	c'_{x,i}(X) = \left| \left \{(f_1,\ldots,f_i) \in F^i: N(\{f_1,\ldots,f_i\}) = X \wedge \sum_{l=1}^i |N(f_l)| = x \right\} \right|,
\]
which is easily seen to be the number $i$-tuples of disjoint sets whose union is exactly $X$ if $x=|X|$. Note that $h_{x,i}(X) = \sum_{Y \subseteq X} c'_{x,i}(Y)$ since every $i$-tuple counted in $h_{x,i}(X)$ is counted once in a $c'_{x,i}(Y)$ where $Y$ is the union of the $i$-tuple. Then, since any non-empty set has equally many even-sized as odd-sized subsets (e.g., by inclusion exclusion), we see that
\begin{equation}\label{eq:incexc}
\begin{aligned}
	c'_{x,i}(X) &= \sum_{Y \subseteq X} \left( \sum_{Z \subseteq X \setminus Y} (-1)^{|Z|} \right) c'_{x,i}(Y)\\
	&= \sum_{Z \subseteq X} (-1)^{|Z|} \left(\sum_{Y \subseteq X \setminus Z} c'_{x,i}(Y)\right) = \sum_{Z \subseteq X} (-1)^{|Z|} h_{x,i}(X\setminus Z).
\end{aligned}
\end{equation}
Then, define
\[
	h^j_{x,i}(X) = \sum_{X \cap\{u_1,\ldots,u_j\} \subseteq Z \subseteq X}(-1)^{|Z|}h_{x,i}(X \setminus Z),
\]
and similarly to~\eqref{eq:yates}, we see that $h^n_{x,i}(X)=h_{x,i}(X)$ and for $j < n$ we can compute $h^j_{x,i}$ using
\begin{equation}\label{eq:moeb}
	h^j_{x,i}[X] = -[u_i \in X]h^{j+1}_{x,i}(X\setminus \{u_i\})+h^{j+1}_{x,i}(X).
\end{equation}

Consequently, we can use Equation~\ref{eq:moeb} for a straightforward dynamic programming algorithm to determine $h^0_{|X|,i}[X]$ which equals $c'_{x,i}(X)$ by~\eqref{eq:incexc}, and it is easy to see that $c_i(X)$ is true if and only if $c_{|X|,i}(X) >0$.

\end{document}